\documentclass[conference,letterpaper]{IEEEtran}

\usepackage[utf8]{inputenc} 
\usepackage[T1]{fontenc}
\usepackage{url}
\usepackage{ifthen}
\usepackage{cite}
\usepackage[cmex10]{amsmath} 
\usepackage{hyperref}

\usepackage{amsfonts, amssymb, amsthm}
\usepackage{array}
\usepackage{tikz, color, soul}
\usetikzlibrary{plotmarks}
\usepackage{mathtools}

\usepackage{mathtools}

\newtheorem{theorem}{Theorem}
\newtheorem{proposition}{Proposition}
\newtheorem{lemma}[theorem]{Lemma}

\begin{document}
\title{A New Algorithm for Two-Stage Group Testing} 

\author{%
   \IEEEauthorblockN{Ilya~Vorobyev}
   \IEEEauthorblockA{
    Center for Computational and Data-Intensive Science and Engineering, \\
    Skolkovo Institute of Science and Technology\\
    Moscow, Russia 127051}
    \IEEEauthorblockA{
    Advanced Combinatorics and Complex Networks Lab, \\
    Moscow Institute of Physics and Technology\\  Dolgoprudny, Russia 141701}
    \IEEEauthorblockA{\textbf{Email}: vorobyev.i.v@yandex.ru}
 }

\maketitle

\begin{abstract}
   Group testing is a well-known search problem that consists in detecting of $s$
defective members of a set of $t$ samples by carrying out tests on properly chosen subsets
of samples. In classical group testing the goal is to find all defective elements by
using the minimal possible number of tests in the worst case. In this work, two-stage group testing is considered. Using the hypergraph approach we design a new search algorithm, which allows improving the known results for fixed $s$ and $t\to\infty$. For the case $s=2$ this algorithm achieves information-theoretic lower bound $2\log_2t(1+o(1))$ on the number of tests in the worst case. Also, the problem of finding $m$ out of $s$ defectives is considered.
\end{abstract}


\section{Introduction}\label{sec::intro}

Group testing problem was introduced by Dorfman in~\cite{dorfman1943detection}. Suppose that we have a population of $t$ items(samples), some of which are defective. Our task is to find all defective items by performing a minimal number of tests.  The test is carried out on a properly chosen subset (pool) of the set of samples. The test outcome is positive if the tested set contains at least one defective element; otherwise, it is negative.  In this work consider the noiseless case, i.e., the outcomes are always correct.

In group testing, two types of algorithms are usually considered. In \textit{adaptive}
group testing, at each step the algorithm decides which group to test by observing
the responses of the previous tests. In \textit{non-adaptive} algorithm, all tests are carried out in parallel. 
\textit{Multistage} algorithm is a compromise solution to the group testing problem. In $p$-stage algorithms all tests are divided into $p$ stages. Tests from the $i$th stage may depend on the outcomes of the tests from the previous stages. 

Define $N_p(t,s)$ to be the minimal worst-case total number of tests needed to find all $s$ defective members of a set of $t$ samples using at most $p$ stages. Also define the optimal \textit{rate} of $p$-stage search algorithm as
$$
R_p(s)=\varliminf_{t\to \infty}\frac{\log_2t}{N_p(t, s)}.
$$
By the similar way we define the rate $R_{ad}(s)$ of fully adaptive algorithms.

In many applications, it is much cheaper and faster to perform tests in parallel, but non-adaptive algorithms require far more tests than adaptive ones.
More precise, for non-adaptive algorithms it is known~\cite{d1982bounds,ruszinko1994upper} that 
${R_1(s) = O(\log_2s/s^2)}$. In contrast, adaptive algorithms allow to achieve the rate ${R_{ad}=1/s}$.
Rather surprisingly, for 2-stage algorithms it was proved that ${O(s \log_2 t)}$ tests are already sufficient~\cite{de2005optimal,rashad1990random,d2014lectures}. This fact emphasizes the importance of multistage algorithms.

\subsection{Previous results}

We refer the reader to the monographs \cite{du2000combinatorial,Cicalese2013Fault-TolerantAlgorithms} for a survey on group testing and its applications. In this paper, only the number of test needed in the worst-case scenario is considered. For the problem of finding the average number of tests we refer the reader to~\cite{freidlina1975design} for $s=O(1)$ and to~\cite{mezard2011group, johnson2019performance} for $s\to\infty$.

For non-adaptive algorithms the best known asymptotic ($s\to\infty$) lower~\cite{d2014lectures} and upper~\cite{d2014bounds} bounds are as follows
$$
\frac{2\ln 2}{s^2}(1+o(1))\leq R_1(s)\leq \frac{4\log_2s}{s^2}(1+o(1)).
$$

In addition, we refer to the work \cite{coppersmith1998new}, where the best lower and upper bounds on $R_1(2)$ were established
$$
0.31349\leq R_1(2)\leq 0.4998.
$$

For the case of $p$-stage algorithms, $p>1$, the only known upper bound is information-theoretic one
\begin{equation}\label{eq::th_inf_bound}
R_p(s)\leq \frac{1}{s}, \quad p>1.
\end{equation}

Group testing algorithms with 2-stages can be constructed from disjunctive list-decoding codes~\cite{dyachkov1983survey} and selectors~\cite{de2005optimal}. Both approaches provide the bound $R_2(s)=\Omega(1/s)$, but best results for disjunctive list-decoding codes give a better constant~\cite{d2014bounds}
\begin{equation}
R_2(s)\geq \frac{\log_2e}{es}(1+o(1)).
\end{equation}

For the specific case $s=2$ the best result was obtained in~\cite{damaschke2014strict}
$$
R_2(2)\geq 0.4098.
$$

All of the lower bounds mentioned above were probabilistic. We want to refer to 2 constructive lower bounds for the case $s=2$. In~\cite{damaschke2013toolbox} the authors obtained a 2-stage algorithm with rate $0.4$. In~\cite{d2016hypergraph} an explicit 4-stage testing scheme with the rate $0.5$ was constructed. This bound matches the information-theoretic upper bound, i.e. the presented 4-stage algorithms allow to achieve the same rate as a fully adaptive algorithm.

The aim of this work is a further development of the bounds on the rates $R_2(s)$.

\subsection{Outline}
In Section~\ref{sec::pre}, we introduce the notation and describe the hypergraph approach to group testing problem. In Section~\ref{sec::2def}, we establish Theorem~\ref{th::2def} which states $R_2(2)=0.5$.
This result means that 2-stage testing schemes can achieve the same rate as fully adaptive algorithms for $s=2$.
Theorem~\ref{th::sDef} proved in Section~\ref{sec::sDef} is a generalization of Theorem~\ref{th::2def} for the case of an arbitrary number of defectives. Numerical results and comparison with the best previously known bounds are presented in Table~\ref{table::comp}.
In section~\ref{sec::m_out_of_s}, we consider the problem of finding  $m$ out of $s$ defective elements. Theorem~\ref{th::m_out_of_s} shows that we can find $[s/2]+1$ defective elements with the rate $1/s$. Finally, Section~\ref{sec::conclusion} concludes the paper.

\section{Preliminaries}\label{sec::pre}

Throughout the paper we use $t$ and $s$ for the number of elements and defectives, respectively. By $[t]$ we denote the set ${\{1, 2\ldots, t\}}$. The binary entropy function $h(x)$ is defined as usual $$h(x)=-x\log_2(x)-(1-x)\log_2(1-x).$$

A binary $(N \times t)$-matrix with $N$ rows $x_1, \dots, x_N$ and $t$ columns $x(1), \dots, x(t)$ 
$$
X = \| x_i(j) \|, \quad x_i(j) = 0, 1, \quad i \in [N],\,j \in [t]
$$
is called a {\em binary code of length $N$  and size $t$}.
The number of $1$'s in the codeword $x(j)$, i.e., $|x(j)| = \sum\limits_{i = 1}^N \, x_i(j)= wN$,
is called the {\em weight} of $x(j)$, $j \in [t]$ and parameter $w$, $0<w<1$, is the \textit{relative weight}.
The quantity $R=\frac{\log_2t}{N}$ is called the rate of the code $X$.

Represent $N$ non-adaptive tests with a binary $N\times t$ matrix $X=\|x_{i,j}\|$ in the following way. 
 An entry $x_{i,j}$ equal $1$ if and only if $j$th element is included in $i$th test.
Let $u \bigvee v$ denote the disjunctive sum of binary columns $u, v \in \{0, 1\}^N$. 
For any subset $S\subset[t]$ define the binary vector $$r(X,S) = \bigvee\limits_{j\in S}x(j),$$
which later will be called the \textit{outcome vector}.
By $S_{un}$, ${|S_{un}|=s}$, denote an unknown set of defects. 

In the sequel, we consider 2-stage search algorithms. During the first stage, some pools are tested in parallel. Tests for the second stage depend on the outcomes of the first stage. 

Let us describe the hypergraph approach to group testing problem. Suppose that we use a binary $N\times t$ matrix $X$ at the first stage. As a result of performed tests we get the outcome vector ${y=r(X, S_{un})}$. Construct a hypergraph ${H(X, s, y)=(V, E)}$ in the following way. The set of vertexes $V$ coincides with the set of samples $[t]$. The set of edges consists of all sets $S\subset[t]$, $|S|=s$, such that $r(X, S)=y$. In other words, the set of edges of the hypergraph $H(X, s, y)$ represents all possible defective sets of size $s$. We want to design such a matrix $X$ for the first stage of an algorithm that the hypergraph $H(X, s, y)$ has some good properties, which will allow us to quickly find all defectives at the next few stages.

We can describe previously known algorithms using this terminology. Disjunctive list-decoding codes and selectors give a binary matrix $X$ such that the hypergraph $H(X, s, y)$ has only a constant amount of edges for all possible outcome vectors $y$. Then we can test all non-isolated vertices individually at the second stage. In the algorithm from~\cite{damaschke2013toolbox} the degree of all vertices of the graph $H(X, 2, y)$ is at most 1. This fact allows dividing vertices into 2 parts, each containing exactly one defective. In the algorithm from~\cite{d2016hypergraph} the graph $H(X, 2, y)$ has a small chromatic number, which also allows finding defectives quickly.

We design a new sufficient condition on matrix $X$ to guarantee that the hypergraph $H(X, s, y)$ has a constant amount of edges. Our condition is weaker than conditions for selectors or disjunctive list-decoding codes and allows to construct matrices with a higher rate. A step-by-step description of our algorithm is as follows:
\begin{enumerate}
\item Take a binary matrix $X$ such that the hypergraph $H(X, s, y)$ has only a constant amount of edges for all possible outcome vectors $y$. Use $X$ as a testing matrix at the first stage.
\item Using the outcome vector $y$ construct the hypergraph $H(X, s, y)$.
\item Test all non-isolated vertices of the hypergraph $H(X, s, y)$ at the second stage.
\end{enumerate}

We perform only a constant amount of tests at the second stage, therefore, the asymptotic($t\to\infty$) rate of such scheme is equal to the asymptotic rate of the code $X$. Lower bounds on the rate of such codes are derived in Theorem~\ref{th::2def} for $s=2$ and in Theorem~\ref{th::sDef} for $s>2$.

\section{Algorithm for 2 defectives}\label{sec::2def}
We apply the algorithm described in the previous section to the case $s=2$. Note that in this case, we have a graph instead of hypergraph $H$. 

\begin{theorem}\label{th::2def}
    $R_2(2)=0.5.$
\end{theorem}
As it was mentioned before, this bound matches the bound for a fully adaptive algorithm. In addition, the algorithm from this theorem can be used to find not only $2$ defectives but also at most $2$ defectives. But to keep things simple, here we consider only the case of exactly 2 defectives.

    \begin{proof}[Proof of Theorem~\ref{th::2def}]
    Consider a random matrix $X$ of size $N\times t$, each column of which is chosen independently and uniformly from the set off all columns of weight $wN$, $0<w<1$. To keep the notation simple, we ignore the fact that $wN$ may not be an integer. Fix the constant $L$ and consider a graph $G_y=H(X, 2, y)$, where vector $y$ is from $\{0,\;1\}^N$.
    
    Call the index $v\in [t]$ a $y$-bad index of the first type if the degree of the vertex $v\in V$ in the graph $G_y$ is at least $L$. 
    Call the index $v\in [t]$ a $y$-bad index of the second type if in the graph $G_y$ the vertex $v\in V$ is included in some matching, which contains at least $L$ edges. Recall that matching is a set of edges without common vertexes.
   Finally, call the index $v\in [t]$ a bad index if there exists a vector $y\in \{0,\;1\}^N$ such that $v$ is a $y$-bad index of the first or the second type.
   
    The following two propositions imply the theorem.
    \begin{proposition}\label{pr::graph}
    If the maximum vertex degree and the maximum cardinality of a matching in a graph $G=(V, E)$ are less than  $L$, then ${|E|< 2L^2}$.
    \end{proposition}    
    \begin{proposition}\label{pr::badVertices}
    For any $R<0.5$ and $w=1-\sqrt{2}/2$, there exists an integer $L$ such that for $t=\lfloor2^{RN}\rfloor$ the mathematical expectation of the number of bad indexes less than $1$ for $N$ big enough.
    \end{proposition}
    Let us show how the Theorem~\ref{th::2def} can be deduced from these two propositions. Indeed, take the parameters $w$, $R$, $L$ from Proposition~\ref{pr::badVertices}. Thus, for $N$ big enough there exists a $N\times t$ matrix $X$ without bad indexes. 
It means that for any outcome vector $y=r(X, S_{un})$ graph $G_y=H(X, 2, y)$ does not contain a matching of size at least $L$ or a vertex with a degree at least $L$. Applying the Proposition~\ref{pr::graph} to the graph $G_y$, we conclude that it has less than $2L^2$ edges. Using the matrix $X$ as a testing matrix at the first stage we obtain a search procedure with asymptotic rate at least $R$.

Now we prove propositions~\ref{pr::graph} and~\ref{pr::badVertices}.

\begin{proof}[Proof of Proposition~\ref{pr::graph}]
Fix an arbitrary maximum matching $M\subset E$, $|M|< L$, in the graph $G=(V, E)$. Denote the set of endpoints of $M$ as $U\subset V$, $|U|< 2L$. Since $M$ is a maximum matching, every edge $e$ has at least one endpoint in the set $U$. Therefore, the total number of edges is upper bounded by 
$$
\sum_{v\in U}deg(v)< 2L^2.
$$
\end{proof}

\begin{proof}[Proof of Proposition~\ref{pr::badVertices}]
Denote the event that a fixed index $v$ is a $y$-bad index of the first(second) type for some outcome vector $y$ as $B_{v, y, 1}$($B_{v, y, 2}$).
The probability $\Pr(B_{v, y, 1})$ for a vector $y$ of weight $qN$ can be upper bounded by the probability that there exists a non-ordered collection of $L$ other vertices $v_1$, $v_2$, \ldots, $v_{L}$, such that the graph $G_y$ contains edges $(v, v_i)$ for $1\leq i \leq L$, therefore,
\begin{equation}\label{eq::bad1}
\Pr(v, y, 1)\leq \binom{t-1}{L}p_1^{L}<t^{L}p_1^{L},
\end{equation}
where $p_1=\binom{wN}{(q-w)N}/\binom{N}{wN}$ is a probability that for some index $u\ne v$ the equation $x(v)\bigvee x(u)=y$ holds.

The probability $\Pr(B_{v, y, 2})$ for a vector $y$ of weight $qN$ is at most the probability that there exists an ordered collection of $2L-1$ other vertices $v_1$, $v_2$, \ldots, $v_{2L-1}$, such that the graph $G_y$ contains edges $(v_{2i}, v_{2i+1})$ for $1\leq i \leq L-1$, and an edge $(v, v_1)$. Hence,
\begin{equation}\label{eq::bad2}
\Pr(v, y, 2)\leq \binom{t-1}{2L-1}(2L-1)!p_1p_2^{L-1}< t^{2L-1} p_2^{L-1},
\end{equation}
where $p_2=\frac{\binom{qN}{wN}\binom{wN}{(q-w)N}}{\left(\binom{N}{wN}\right)^2}$ is a probability that for some indexes $u_1, u_2$ the equation $x(u_1)\bigvee x(u_2)=y$ holds true.

Therefore, the mathematical expectation of the number of bad indexes can be upper bounded as follows
\begin{multline}
t2^{N} \sup\limits_{q\in (w,\;\min(2w,1))}(t^{L} p_1^{L}+t^{2L-1} p_2^{L-1})\\
< 2^N \sup\limits_{q\in (w,\;\min(2w,1))}(2^{RN(L+1)}p_1^L+2^{RN(2L+2)}p_2^L).
\end{multline}

Take $L$ such that $R(L+1)-0.5L=-1-\varepsilon$, $\varepsilon>0$. Then the mathematical expectation of the number of bad indexes is less than
$$
2^{-\varepsilon N} \sup\limits_{q\in (w,\;\min(2w,1))}((2^{0.5N}p_1)^L+(2^{N}p_2)^L).
$$
To finish the proof of the proposition it is sufficient to show that ${p_1<2^{-0.5N(1+o(1))}}$ and ${p_2<2^{-N(1+o(1))}}$ for all $q$. It is easy to see that ${p_1^2\leq p_2}$, hence it is enough to verify the inequality ${p_2<2^{-N(1+o(1))}}$.

Taking the logarithm and dividing by $N$, we obtain
\begin{multline}
 \sup\limits_{q\in (w,\;\min(2w,1))}qh(w/q)+wh((q-w)/w)-2h(w)\\
<-1+o(1).
\end{multline}
For $w=1-\sqrt{2}/2$ the maximal value of the left-hand side is equal to $-1$, therefore, the inequality holds.

\end{proof}
The theorem is proved.
\end{proof}

\section{Algorithm for $s$ defectives}\label{sec::sDef}
To construct a matrix for the first stage of our algorithm for $s>2$  we must introduce some new notions. 
Fix an integer $L$
and consider a $s$-uniform hypergraph $H$. Call the set of edges $e_1, e_2, \ldots, e_L$ a $(s, L, k)$-\textit{bad configuration} if $e_i\cap e_j=U$, $|U|=k$, for any $i$ and $j$. In other words, $(s, L, k)$-bad configuration consists of $L$ edges such that the intersection of every two edges is the same set of size $k$. 
Call a code $X$ a \textit{$(s, L, K)$-good} code, $K\subset\{0, 1,\ldots, s-1\}$, if the hypergraph $H(X, s, y)$ doesn't contain a $(s, L, k)$-bad configuration for any outcome vector $y$ and integer $k\in K$.
Let $N(t, s, L, K)$ be the minimal length of $(s, L, K)$-good code of size $t$. The asymptotic rate $R(s, L, K)$ of $(s, L, K)$-good code is defined as follows
\begin{equation}
R(s, L, K)=\varliminf_{t\to\infty} \frac{\log_2t}{N(t, s, L, K)}.
\end{equation}
Denote the limit $\lim\limits_{L\to\infty}R(s, L, K)$ by $R(s, \infty, K)$.

The following lemma demonstrates the connection between $(s, L, K)$-good codes and two-stage group testing problem.
\begin{lemma}\label{lem::sDef}
\begin{equation}\label{eq::lemma}
R_2(s)\geq R(s, \infty, \{0, 1, \ldots, s - 1\}).
\end{equation}
\end{lemma}

\begin{proof}[Proof of Lemma~\ref{lem::sDef}]
Use $(s, L, \{0, 1, \ldots, s - 1\})$-good code $X$ of size $t$ as a test matrix at the first stage. Then for any outcome vector $y$ hypergraph $H(X, s, y)$ doesn't contain $(s, L, k)$-bad configurations for $k=0, 1, \ldots s-1$. 

\begin{proposition}\label{pr::hypergraph}
If a $s$-uniform hypergraph $H=(V, E)$ doesn't contain $(s, L, k)$-bad configurations for $k=0, 1, \ldots s-1$, then the number of edges $|E|$ is at most $c(s, L)$, where $c(s, L)$ doesn't depend on $|V|$.
\end{proposition}
\begin{proof}[Proof of Proposition~\ref{pr::hypergraph}]

Suppose, seeking a contradiction, that a hypergraph $H=(V, E)$ without bad configurations contains more than $c(s, L)$ edges. Exact formula for $c(s, L)$ will be specified later. Construct a complete graph $\hat{G}=K_{|E|}$, which vertex set {$V=\{e_1, \ldots, e_{|E|}\}$} corresponds to the edges of hypergraph $H$. Color the edge $f=(e_1, e_2)$ of the graph $\hat{G}$ in color $i+1$, if the cardinality of the intersection $e_1\cap e_2$ is equal to $i$, $i=0,1,\ldots, s-1$.

Recall that Ramsey number $R(c_1, \ldots, c_l)$ is a minimal integer $n$ such that if the edges of a complete graph $K_n$ are colored with $l$ different colors, then for some $i$ between $1$ and $c$, the graph must contain a complete subgraph of size $c_i$ whose edges are all color $i$.  Here we need only the fact that the number $R(c_1, \ldots, c_l)$ exists.

    Take $c(s, L)=R(\underbrace{c_0(s, L), \ldots, c_0(s, L)}_{s})$. Then for some $k$ there exists a set $E_0$ of edges $e_1$, \ldots, $e_{c_0(s,L)}$ from the hypergraph $H$ such that $|e_i\cap e_j|=k$ for any ${1\leq i<j\leq c_0(s,L)}$.
    
    Consider an edge $e_1$. Any other edge from the set $E_0$ has $k$ common vertexes  with $e_1$. Taking $c_0(s,L)>\binom{s}{k}(L-1)$, we obtain that some $k$ vertexes $v_1, \ldots, v_k$ belong to another $L-1$ edges $w_1, \ldots w_{L-1}$ from the set $E_0$. But then the set of edges $e_1, w_1, \ldots, w_{L-1}$ forms bad configuration of type $k$. This contradiction proves the proposition.
\end{proof}
Using Proposition~\ref{pr::hypergraph} we conclude that the hypergraph $H(X, s, y)$ has at most $c(s, L)$ edges. Thus, we can find all defectives by testing all non-isolated vertices individually at the second stage. The number of tests at the second stage doesn't depend on the number of elements $t$, therefore, taking limits $t\to\infty$ and $L\to\infty$ we obtain the inequality~\eqref{eq::lemma}.

\end{proof}

To obtain a lower bound on the rate $R(s, \infty, K)$ we use a random coding method.

\begin{lemma}\label{lem::goodCodes}
Define a function $A(s, w, q)$ 
\begin{multline}
A(s, w, q) = (1-q)\log_2(1-q)+
q\log_2\left(\frac{wy^s}{1-y}\right)\\
+sw\log_2\frac{1-y}{y}+sh(w), \quad w<q<\min(1, sw),
\label{eq::functionA}
\end{multline}
where $y\in(0,\;1)$ is a unique root of the equation
\begin{equation}
q=w\frac{1-y^s}{1-y}.\label{eq::y}
\end{equation}

Define $\underline{R}(s, k, w)$ as follows.
\begin{equation}
\underline{R}(s, k , w)=\min\left\{\underline{R}_1(s, k , w), \underline{R}_2(s, k , w)\right\},
\end{equation}
where 
\begin{equation}\label{thCond1}
\underline{R}_1(s, k , w)=\inf\limits_
{\substack{\max(w, kw/2)\leq q\\
q \leq \min((s-k)w,1)}}
\frac{A(s-k, w, q)-kw+h(q)}{s-k},
\end{equation}
\begin{equation}\label{thCond2}
\underline{R}_2(s, k , w)=\inf\limits_
{\substack{w\leq q\leq1\\
q\leq (s-k)w\\
q\leq kw/2}}
\frac{A(s-k, w, q)-kwh\left(\frac{q}{kw}\right)+h(q)}{s-k}.
\end{equation}
Then
\begin{equation}
R(s, \infty, K)\geq \sup\limits_{0<w<1}\min\limits_{k\in K}\underline{R}(s, k , w)
\end{equation}
\end{lemma}
The proof of this lemma can be found in the Appendix.

Lemma~\ref{lem::goodCodes} and Lemma~\ref{lem::sDef} give us
\begin{theorem}\label{th::sDef}
\begin{equation}
R_2(s)\geq \sup\limits_{0<w<1}\min\limits_{0\leq k < s}\underline{R}(s, k , w)
\end{equation}
\end{theorem}

The best previously known lower bounds for the case $s>2$ are given by disjunctive list-decoding codes with the length of the list ${L\to\infty}$~\cite{d2014bounds}. In Table~\ref{table::comp} we compare bounds given by Theorem~\ref{th::sDef} with the best previously known lower bounds.

\begin{table}[htbp]
\renewcommand{\arraystretch}{1.3}
\caption{Comparison of old and new lower bounds on the rate $R_2(s)$}
\label{table::comp}
\begin{center}
\begin{tabular}{ ||c||c|c|c|c|c|| } 
\hline
$s$ & 3 & 4 & 5 & 6\\
\hline
\text{old} & 0.199 & 0.145 & 0.114 & 0.094\\
\hline
\text{new} & 0.3219 & 0.199 & 0.145 & 0.114\\
\hline
\end{tabular}
\end{center}
\end{table}

Note that the new lower bound for $s+1$ defective elements coincides with the old lower bound for $s$ defective elements. It is easy to show that $(s, L, \{s-1\})$-good code is a $s$-disjunctive list-decoding code with a list of size $L$. Therefore, a new bound for $s$ defectives can't be better than an old bound for $s+1$ defectives. In particular, it means that a new algorithm doesn't improve the previously best known bound for ${s\to\infty}$.

\section{Finding $m$ out of $s$ defectives}\label{sec::m_out_of_s}
The technique developed in the previous sections can be used to find only part of the defectives. Suppose that we want to find only $m$ out of $s$ defectives.  Define $N_p(t,s,m)$ to be the minimal worst-case total number of tests needed to find $m$ out of $s$ defective members in a set of $t$ samples using at most $p$ stages. Also define the optimal \textit{rate} of $p$-stage search algorithm as
$$
R_p(s, m)=\varliminf_{t\to \infty}\frac{\log_2t}{N_p(t, s, m)}.
$$
The problem of finding $m$ out of $s$ defectives with the help of non-adaptive algorithms was formulated in~\cite{csros2005single} for $m=1$ and in~\cite{laczay2006multiple} for the general case. The best results were obtained in~\cite{alon2006tracing,alon2007tracing}, where the following bound was proved
$$
R(s, m)\geq \min\left(\frac{c_1}{s}, \frac{c_2}{m^2}\right)
$$
for some constants $c_1$ and $c_2$.

For two-stage algorithms, we don't know if it is possible to find all defectives with the rate $1/s$. But it turns out that we can find at least half of the defectives with this rate.

\begin{theorem}\label{th::m_out_of_s}
\begin{equation}\label{eq::m_out_of_s}
R_2(s, \lfloor s/2\rfloor+1)\geq\frac{1}{s}.
\end{equation}
\end{theorem}
\begin{proof}[Proof of Theorem~\ref{th::m_out_of_s}]
This proof is based on the following technical lemma.
\begin{lemma}\label{lem::m_out_of_s}
\begin{equation}
R(s, \infty, \{0, 1, \ldots, \lfloor s/2\rfloor\})\geq\frac{1}{s}.
\end{equation}
\end{lemma}
The proof of this lemma is postponed to the Appendix. Let us show how it implies the theorem.
Use $(s, L,  \{0, 1, \ldots, \lfloor s/2\rfloor\})$-good code $X$ of size $t$ as a test matrix at the first stage. Then for any outcome vector $y$ hypergraph $H(X, s, y)=(V, E)$ doesn't contain $(s, L, k)$-bad configurations for $k=0, 1, \ldots \lfloor s/2\rfloor$. Form subset of edges $E_1\subset E$ such that for every $e_1, e_2\in E_1$ the intersection $e_1\cap e_2$ has at most $\lfloor s/2\rfloor$ vertices, and for every $e_1\in E_1$, $e_2\in E\setminus E_1$ the intersection $e_1\cap e_2$ has at least $\lfloor s/2\rfloor + 1$ vertices. Such subset can be constructed greedily by adding edges one by one while it is possible. At the second stage we test all non-isolated vertices of the hypergraph $H_1=(V, E_1)$. At least $\lfloor s/2\rfloor+1$ vertices of every edge $e\in E$ will be tested by construction of $E_1$, thus, we will find at lest $\lfloor s/2\rfloor+1$ defectives.
The number of tests at the second stage is upper bounded by $s\|E_1\|$. 
The hypergraph $H_1=(V, E_1)$ doesn't contain any $(s, L, k)$-bad configurations for $k=0, 1, \ldots, s-1$. Using Proposition~\ref{pr::hypergraph} we conclude that $|E_1|\leq c(s, L)$. Therefore, the number of tests at the second stage doesn't depend on the total number of elements $t$. Taking limits $t\to\infty$ and $L\to\infty$ and using lemma~\ref{lem::m_out_of_s} we obtain the desired inequality~\eqref{eq::m_out_of_s}.
\end{proof}

\section{Conclusion}\label{sec::conclusion}

A new algorithm for two-stage group testing was proposed, which improves previously known results. For the case of 2 defectives, this algorithm has the optimal rate 0.5. Also, a two-stage algorithm which finds at least half of the defectives with the rate $1/s$ was constructed. 

Development of the algorithm, which will achieve the optimal rate for the number of defectives greater than 2, is a natural open problem. Another interesting task is to obtain an upper bound on the rate $R_p(s)$, $p>1$, which is stronger than information-theoretic bound $1/s$.

We note that the technique used in this paper could be also applied to other group testing models, such as, for example, symmetric or threshold group testing.


\section{Acknowledgement}
I. Vorobyev was supported in part by RFBR through grant nos. 18-07-01427~A, 18-31-00361~MOL\_A.
    

\bibliographystyle{IEEEtran}
\bibliography{2StageGT}
\appendix
\begin{proof}[Proof of Lemma~\ref{lem::goodCodes}]
Consider a random matrix $X$ of size $N\times t$, $t=2^{\lfloor RN\rfloor}$, each column of which is chosen independently and uniformly from the set of all columns of weight $wN$, $0<w<1$. To keep the notation simple we ignore the fact that $wN$ may not be an integer. We want to prove that with a positive probability hypergraph $H(X, s, y)$ doesn't contain bad configurations for any outcome vector ${y\in\{0,\;1\}^N}$.


Estimate the mathematical expectation of the number of bad configurations $\xi_k$ for $k=0,\ldots,s-1$. Fix outcome vector $y$ of weight $qN$. Denote the weight of the union of $k$ common vertexes of configuration as $q_0N$, the weight of the union of additional $s-k$ vertexes as $q_1N$.
Let $P(l, q, Q)$ be equal to the probability that the weight of the union of $l$ random columns equals $Q$, where each column is chosen independently and uniformly from the set of all columns of weight $q$. Then the following inequality holds true. 
\begin{multline*}
E \xi_k
\leq 2^{N}
\max\limits_{I(k>0)w\leq q_0\leq kw} 
t^kN\binom{N}{q_0N}
\frac{P(k, wN, q_0N)}{\binom{N}{q_0N}}\\
\times\max\limits_{\substack{q_0<q<sw,\\w<q_1<(s-k)w}} 
\left(t^{s-k}N\binom{q_0N}{(q-q_1)N}\frac{P(s-k, wN, q_1N)}{\binom{N}{q_1N}}\right)^L,
\end{multline*}
where $\frac{P(\ell, wN, q_iN)}{\binom{N}{q_iN}}$ is a probability that the union of $\ell$ columns gives a specific column of weight $q_iN$, $t^k\times t^{(s-k)L}$ is an upper bound for the number of ways to choose indexes of columns, $2^{N}\times N\binom{N}{q_0N}\times N\binom{q_0N}{(q-q_1)N}$ is an upper bound for the number of ways to choose an outcome vector $y$, a weight $q_0N$, a vector of weight $q_0N$, a weight $q_1N$ and a vector of weight $q_1N$ such that disjunctive union of this vector with a fixed vector of weight $q_0N$ gives a fixed vector of weight $qN$. 
Let the function $A(l, q, Q)$ be defined by
$$
A(l, q, Q) = -\lim\limits_{N\to\infty}\frac{\log_2P(l, q, Q)}{N}.
$$
We use the representation~\eqref{eq::functionA}-\eqref{eq::y} of the function $A(l, q, Q)$, which was established in Theorem~4 of paper~\cite{d2015almost}.

We want to prove that $E\xi_k<1/s$. Taking the logarithm and dividing both parts by $NL$, then taking a limit $\lim\limits_{L\to\infty}\lim\limits_{N\to\infty}$, we obtain a sufficient condition for the inequality $E\xi_k<1/s$.
\begin{multline}\label{ineq}
\sup\limits_{*}(s-k)R-A(s-k, w, q_1)\\
+q_0h\left(\frac{q-q_1}{q_0}\right)-h(q_1)\leq o(1),
\end{multline}
\begin{multline}
*=I(k>0)w\leq q_0\leq kw, \max(q_0, q_1)\leq q\leq \min(sw, 1) \\
w\leq q_1\leq (s-k)w,
\end{multline}
where $w_0h\left(\frac{w-w_1}{w_0}\right)$ is defined to be equal to $0$ at $q_0=0$. Note that left-hand side of~\eqref{ineq} is an increasing function of $q_0$, therefore, either $q_0=kw$ or $q_0=q$.

\begin{enumerate}
\item
If $q_0=kw$, then we obtain the condition
\begin{multline}\label{cond0}
\sup\limits_{*}(s-k)R-A(s-k, w, q_1)\\
+kwh\left(\frac{q-q_1}{kw}\right)-h(q_1)\leq o(1),
\end{multline}
$$
*=\max(kw, q_1)\leq q\leq \min(sw, 1) w\leq q_1\leq (s-k)w.
$$
We have to consider two more cases.

\begin{enumerate}
\item
Case $q_1\leq kw/2$. Maximum is attained at $q=kw$. This leads to 
\begin{multline}\label{cond1}
\sup\limits_{w\leq q_1\leq \min((s-k)w, kw/2)}(s-k)R-A(s-k, w, q_1)\\
+kwh\left(\frac{q_1}{kw}\right)-h(q_1)\leq o(1),
\end{multline}
which is equivalent to $R\leq \underline{R}_2(s, k , w)$, where $\underline{R}_2(s, k , w)$ is defined in~\eqref{thCond2}.
\item
Case $q_1\geq kw/2$. In this case optimal $q$ is equal to $q_1+kw/2$. Then condition~\eqref{cond0} transforms into
\begin{equation}\label{cond2}
\sup\limits_{*}
(s-k)R-A(s-k, w, q_1)+kw-h(q_1)\leq o(1),
\end{equation}
$$*=
\max(w, kw/2)\leq q_1\leq (s-k)w,
$$
which is equivalent to $R\leq \underline{R}_1(s, k , w)$, where $\underline{R}_1(s, k , w)$ is defined in~\eqref{thCond1}.
\end{enumerate}
\item
If $q_0=q$, then we obtain
\begin{multline}\label{eq::2}
\sup\limits_{*}
(s-k)R-A(s-k, w, q_1)\\
+qh\left(\frac{q_1}{q}\right)-h(q_1)\leq o(1),
\end{multline}
$$
*=q_1\leq q\leq kw, w\leq q_1\leq (s-k)w.
$$
The left-hand side of~\eqref{eq::2} is an increasing function of $q$, thus we put $q=kw$. This leads to
\begin{multline}\label{cond3}
\sup\limits_{w\leq q_1\leq \min(kw, (s-k)w)}(s-k)R-A(s-k, w, q_1)\\
+kwh\left(\frac{q_1}{kw}\right)-h(q_1)\leq o(1).
\end{multline}
Note that condition~\eqref{cond3} is weaker than condition~\eqref{cond2} for $q_1\geq kw/2$; for $q_1\leq kw/2$ it coincides with condition~\eqref{cond1}. Therefore, it can be omitted.
\end{enumerate}
\end{proof}

\begin{proof}[Proof of Lemma~\ref{lem::m_out_of_s}]

Consider function $f(q)=A(s, w, q) + h(q)$ as a function of $q$, $w\leq q\leq\min(sw, 1)$. 
\begin{proposition}\label{pr::der}
Functions $f(q)$ attains its minimal value at the point $q_{min}=\frac{w}{2\left(1-2^{-\frac{1}{s}}\right)}$.
\end{proposition}
\begin{proof}[Proof of Proposition~\ref{pr::der}]
Using the representation~\eqref{eq::functionA}-\eqref{eq::y} we represent $f(q)$ as a function of $y$. Taking derivative we find that it attains its minimal value at the point $y=2^{-1/s}$. Since there is a bijection between $q$ and $y$, we conclude that $f(q)$ attains its minimal value at the corresponding point $q_{min}$.
\end{proof}
Let $K$ be a set $\{0, 1, \ldots, \lfloor s/2\rfloor\}$. The following chain of inequalities holds.
\begin{multline}
R(s,\infty, K)\geq \sup\limits_{0<w<1}\min\limits_{k\in K}\min(\underline{R}_1(s, k , w), \underline{R}_2(s, k , w))\geq\\
\sup\limits_{0<w<1}\min\limits_{k\in K}\frac{A(s-k, w, q_{min})-kw+h(q_{min})}{s-k}=\\
\sup\limits_{0<w<1}\min\limits_{k\in K}h(w)+w\left(\log_2\left(2^{1/(s-k)}-1\right)-\frac{k}{s-k}\right).
\end{multline}
Function $\log_2\left(2^{1/(s-k)}-1\right)-\frac{k}{s-k}$ is a concave function of $k$, therefore, its minimum is achieved either at $k=0$ or $k=\lfloor s/2\rfloor$. It is easy to check that the minimum is attained at $k=0$. Thus, 
\begin{multline}
R(s,\infty, K)\geq 
\sup\limits_{0<w<1}h(w)+w\left(\log_2\left(2^{1/s}-1\right)\right), 
\end{multline}
which leads to
$$
R(s,\infty, K)\geq \frac{1}{s}
$$
for $w=1-2^{-1/s}$.

\end{proof}

\end{document}